\documentclass[12pt]{article}

\usepackage{enumerate}

\usepackage{a4wide}
\usepackage{amssymb,amsmath, amsthm, amstext, graphics, color}
\usepackage{graphicx}
\usepackage{multicol}
\usepackage{color}
\usepackage[noadjust]{cite}
\usepackage{enumitem}
\usepackage{comment}
\usepackage{pgf,tikz}
\usepackage{calrsfs}
\usetikzlibrary{arrows,shapes}
\usepackage{latexsym, boxedminipage}

\def\C{\mathcal{C}}

\def\F{\mathcal{F}}

\newtheorem{theorem}{Theorem}
\newtheorem{lemma}{Lemma}
\newtheorem{proposition}{Proposition}[section]

\newtheorem{corollary}{Corollary}[section]
\newtheorem{conjecture}{Conjecture}[section]

\newtheorem{claim}{Claim}
\newtheorem{problem}{Problem}

\theoremstyle{definition}

\usepackage{amsthm}

\newcommand{\problemdef}[3]{
	\begin{center}
		\begin{boxedminipage}{.99\textwidth}
			\textsc{{#1}}\\[2pt]
			\begin{tabular}{ r p{0.8\textwidth}}
				\textit{~~~~Instance:} & {#2}\\
				\textit{Question:} & {#3}
			\end{tabular}
		\end{boxedminipage}
	\end{center}
}

\title{Partitioning a graph into degenerate subgraphs}
\author{Faisal N. Abu-Khzam\thanks{Department of Computer Science and Mathematics, Lebanese American University, Beirut, Lebanon \newline email: \texttt{faisal.abukhzam@lau.edu.lb}}, \quad Carl Feghali\thanks{Department of Informatics, University of Bergen, Bergen, Norway. \newline emails: \texttt{\{carl.feghali, pinar.heggernes\}@uib.no}}, \quad Pinar Heggernes\footnotemark[2]}
\date{}

\begin{document}   
\maketitle

\begin{abstract}
Let $G = (V, E)$ be a graph with maximum degree $k\geq 3$ distinct from $K_{k+1}$. Given integers $s \geq 2$ and  $p_1,\ldots,p_s\geq 0$,  $G$ is said to be $(p_1, \dots, p_s)$-partitionable 
if there exists a partition of $V$ into sets~$V_1,\ldots,V_s$ such that $G[V_i]$ is $p_i$-degenerate for $i\in\{1,\ldots,s\}$. In this paper, we prove that we can find a $(p_1, \dots, p_s)$-partition of $G$ in $O(|V| + |E|)$-time whenever  $1\geq p_1, \dots, p_s \geq 0$ and $p_1 + \dots + p_s \geq k - s$.  This generalizes a result of Bonamy et al. (MFCS, 2017) and can be viewed as an algorithmic extension of Brooks' theorem and several results on vertex arboricity of graphs of bounded maximum degree. 
We also prove that deciding whether $G$ is $(p, q)$-partitionable is $\mathbb{NP}$-complete for every $k \geq 5$ and pairs of non-negative integers $(p,  q)$ such that $(p, q) \not = (1, 1)$ and  $p + q = k - 3$. 
 This resolves an open problem of Bonamy et al. (manuscript, 2017). 
Combined with results of Borodin, Kostochka and Toft (\emph{Discrete Mathematics}, 2000), 
Yang and  Yuan (\emph{Discrete Mathematics}, 2006) and  Wu, Yuan and Zhao (\emph{Journal of Mathematical Study}, 1996), 
it also settles the complexity of deciding whether a graph with bounded maximum degree can be partitioned into two subgraphs of prescribed degeneracy. 
\end{abstract}

\section{Introduction}

The concept of degenerate graphs introduced by Lick and White \cite{lick} in 1970 
has since found a number of applications in graph theory, especially in graph partitioning and graph colouring problems. 
This is mainly because the class of degenerate graphs captures one of the earliest studied classes of graphs such as independent sets, forests and planar graphs. 
For example, 
results of Thomassen \cite{Th01, Tho22} on decomposing the vertex set of a planar graph into degenerate subgraphs 
have lead to new proofs of the 5-colour theorem on planar graphs that do not use Euler's formula. Another example is a result of Alon, Kahn and Seymour \cite{alon} that extends the well-known Tur\'an's theorem on the size of the largest independent set in a graph to
the size of largest subgraph of any prescribed degeneracy. 

In this paper, we shall
investigate the complexity of partitioning the vertex set of a graph of bounded
maximum degree into degenerate subgraphs. In order to make this statement more precise, we must first proceed with some definitions. 
 Let $G = (V, E)$ be a graph, 
  and let $k$ be a non-negative integer. We say that $G$ is   $k$-degenerate if we can successively delete vertices of degree at most $k$ in $G$
 until the empty graph is obtained. Expressed in an another way, $G$ is $k$-degenerate
 if it admits a \emph{$k$-degenerate ordering} -- an ordering $x_1, \dots, x_n$ of the vertices in $G$ such that $x_i$ has at most $k$ neighbours $x_j$ in $G$ with $j < i$.  
 In this case, the ordering is said to \emph{start at} $x_1$ and \emph{end at} $x_n$. 
 
  Given integers $s \geq 2$ and  $p_1,\ldots,p_s\geq 0$, $G$ is said to be \emph{$(p_1, \dots, p_s)$-partitionable} if there exists a partition of $V$ into sets~$V_1,\ldots,V_s$ such that $G[V_i]$ is $p_i$-degenerate for $i\in\{1,\ldots,s\}$. 
  
  We shall consider the following computational problem. 
  
  \begin{problem}\label{problem1}
Given a graph $G$ and integers $s \geq 2$, $p_1,\ldots,p_s\geq 0$, determine the complexity of deciding whether $G$ is $(p_1, \dots, p_s)$-partitionable.   
\end{problem}
  
We briefly review some existing results related to Problem~\ref{problem1}. 
Let $G = (V, E)$ be a connected graph with maximum degree $k \geq 3$ distinct from $K_{k+1}$, and
let $d$ be a non-negative integer. A $d$-colouring of $G$ is a function $f: V \rightarrow \{1, \dots, d\}$ such that $f(u) \not=f(v)$ whenever $(u, v) \in E$. 
Equivalently, $f$ is a $d$-colouring of $G$ if  $f^{-1}(1), \dots, f^{-1}(d)$ each forms an independent set. 
 The earliest result on Problem \ref{problem1} is most likely the celebrated theorem of Brooks \cite{brooks}, which states that $G$ has a $d$-colouring for each $d \geq k$. 
Thus, given that an independent set is a $0$-degenerate graph,
 Brooks' theorem can be reformulated in the language of Problem \ref{problem1} to state that $G$ is $(p_1, p_2, \dots, p_s)$-partitionable for every $s \geq k$ and $p_1 = \dots = p_s = 0$.  
Later on, Borodin, Kostochka and Toft \cite{toft} obtained a generalization of Brooks' theorem 
by showing that $G$ remains $(p_1, \dots, p_s)$-partitionable for every $s \geq 2$ and $\sum_{i = 1}^s p_i  \geq k - s$. 
Observe that this result is algorithmic: Given a graph $G$ of maximum degree $k\geq 3$ and integers $s \geq 2$ and  $p_1,\ldots,p_s\geq 0$ such that
$ \sum_{i = 1}^s p_i  \geq k - s$, 
one can check in polynomial time if $G$  is $(p_1, \dots, p_s)$-partitionable, because the only computation needed is to verify whether $G$ is isomorphic to $K_{k+1}$.  
The question remains, however, whether one can \emph{find} such a partition efficiently whenever it exists. 
 In this direction, Bonamy et al.~\cite{bonamy} have already considered the case $s = 2$ with $p_1 = 0$ and $p_1 + p_2 \geq k - 2$ 
by showing that one can find the partition in $O(n + m)$-time if $k = 3$ and in $O(kn^2)$-time if $k \geq 4$. 
In the first part of this paper, we generalize the first of these two results in the following theorem.

 \begin{theorem}\label{thm:forests}
 Let $G$ be a connected graph with maximum degree $k \geq 3$ distinct from $K_{k + 1}$. For every $s \geq 2$ and $1 \geq p_1,\ldots,p_s\geq 0$ such that $\sum_{i = 1}^s p_i  \geq k-s$, a $(p_1, \dots, p_s)$-partition of $G$ can be found in $O(n + m)$-time. 
 \end{theorem}

The proof of Theorem \ref{thm:forests} appears in Section \ref{section:linear}. We remark by our earlier discussion that 
Theorem \ref{thm:forests} can be viewed as an algorithmic extension of Brooks' theorem. 
(In fact, our approach differs from \cite{bonamy} but is instead a refinement of Lov\'asz' proof~\cite{lovasz} of Brooks' theorem -- see \cite{wood} for an algorithmic analysis of~\cite{lovasz}.)
We also remark that since the definition of forests coincides with the definition of 1-degenerate graphs,  
	Theorem \ref{thm:forests} can be viewed as an algorithmic counterpart to several results on vertex arboricity of graphs; see \cite{arboricity1, arboricity2} for some examples.

On a different tack, one might also ask what happens if the maximum degree of the graph exceeds $ s + \sum_{i = 1}^s p_i$. In this direction, Yang and  Yuan~\cite{yang} have shown that the case $s = 2$ with $p_1 = 0$ and  $p_2 = 1$ is $\mathbb{NP}$-complete for every $k \geq 4$. 
 Wu, Yuan and Zhao~\cite{wu} have also shown that the case $s= 2$ with $p_1 = p_2 = 1$ is $\mathbb{NP}$-complete for every~$k \geq 5$. 
 Extending Yang and Yuan's result, Bonamy et al~\cite{bonamybipartite} 
have shown that the case $s = 2$ with $p_1 = 0$ and $p_2 = t - 2$ remains $\mathbb{NP}$-complete for every $t \geq 3$ and  $k \geq 2t - 2$. 
They then posed as an open problem the case $s = 2$ with $p_1 = 0$ and $p_2 = k - 3$ for every $k \geq 5$.   In the second part of this paper, we resolve this problem by proving, more generally, the following theorem.   

 \begin{theorem}\label{thm:complete}
For every integer $k \geq 5$ and pairs of non-negative integers $(p,  q)$ such that $(p, q) \not = (1, 1)$ and  $p + q = k - 3$, 
deciding whether a graph with maximum degree $k$ is $(p, q)$-partitionable
is $\mathbb{NP}$-complete.
\end{theorem}
 
  The proof of Theorem~\ref{thm:complete} appears in Section \ref{section:hard}. 
Let us note that finding the least integer $d$ such that a graph with maximum degree 3 is $d$-colourable can be done in polynomial time. 
Indeed, one can check in polynomial time if $d=1$ or $d=2$ (for any arbitrary graph). If this is not the case, we check in polynomial time if the graph is isomorphic to $K_4$; if not, then we know  
 by Brooks' theorem that $d=3$. Thus, combined with the aforementioned results in \cite{toft, yang, wu}, 
Theorem \ref{thm:complete} settles the complexity of deciding whether a graph with bounded maximum degree can be partitioned 
into two subgraphs of prescribed degeneracy. More formally, we now have the following solution to the case $s = 2$ of Problem \ref{problem1}.
 \begin{corollary}\label{cor:dichotomy}
 Given integers $p, q \geq 0$, deciding whether a graph with maximum degree $k \geq 3$ is $(p, q)$-partitionable is
\begin{itemize}
\item[(i)] polynomial time solvable if $k = 3$ or $p + q \geq k - 2$ or $p = q = 0$;
\item[(ii)] $\mathbb{NP}$-complete otherwise.  
\end{itemize}
 \end{corollary}

\section{A linear time algorithm}\label{section:linear}

In this section, we prove Theorem~\ref{thm:forests}.
 First, we need some standard definitions.

  Let $k$ be a non-negative integer, and let $G$ be a graph with maximum degree $k$. 
Then $G$ is said to be $k$-regular if every vertex of $G$ has degree exactly $k$. A vertex $v$ of $G$ is called a \emph{cut vertex} of $G$ if $G - \{ v\}$ has more components than $G$.   A \emph{block} of $G$ is either $K_2$ or a maximal $2$-connected subgraph of $G$, and an \emph{end block} of $G$ is a block of $G$ that contains exactly one cut vertex of $G$. A \emph{forest partition} of $G$ is a partition of $V$ into $k / 2$ forests if $k$ is even and $\lfloor k / 2 \rfloor$ forests and one independent set if $k$ is odd.

 \begin{lemma}\label{lem:split}
 Let $G$ be a  graph with maximum degree $k \geq 3$. If $G$ has a forest partition that can be found in $O(n + m)$ time, then $G$ has a  $(p_1, \dots, p_s)$-partition for every $1 \geq p_1,\ldots,p_s\geq 0$ such that $\sum_{i = 1}^s p_i \geq k-s$ that can be found in $O(n + m)$ time.  
 \end{lemma}
 
 \begin{proof}
Suppose $k$ is even (the case $k$ is odd is entirely similar) and let $p_1, \dots, p_{s}$ be integers such that $1 \geq p_1,\ldots,p_{s} \geq 0$. 

\medskip

\noindent
\textit{Case 1} $\sum_{i = 1}^s p_i= k-s$. Let $\sigma = \sum_{i = 1}^s(1 - p_i)$, 
then $\sigma$ is even.  Let $\F$ be a forest partition of $G$, and partition $\sigma / 2$ of the forests in $\F$ into two independent sets, which can be done in $O(n + m)$ time.  The resulting decomposition is a $(q_1, \dots, q_t)$-partition of $G$, where $1 \geq q_1,\ldots,q_t\geq 0$. Since $t = |\F| + \sigma / 2 = (k + \sigma) / 2 = s$ and $\sum_{i = 1}^s q_i = \frac{k}{2} - \frac{\sigma}{2} = k - s$, this completes Case 1. 

\medskip

\noindent
\textit{Case 2} $\sum_{i = 1}^s p_i > k-s$. Let $q_1, \dots, q_s$ be integers such that $0 \leq q_i \leq p_i$ for $1 \leq i \leq s$ and such that  $\sum_{i=1}^s q_i = k - s$. By Case 1, a $(q_1, \dots, q_s)$-partition of $G$ can be found in $O(n + m)$ time. This partition is also trivially a $(p_1, \dots, p_s)$-partition of $G$. This completes Case 2 and hence the proof of the lemma. 
 \end{proof}

To  prove Theorem \ref{thm:forests}, it suffices to show by Lemma \ref{lem:split} that for every connected graph $G$ with maximum degree $k \geq 3$ distinct from $K_{k + 1}$, a forest partition of $G$ can be found in $O(n + m)$ time.
We will need the next two lemmas. The proof of the first lemma is essentially the same as the proof of \cite[Lemma 8]{FJP16} but with some minor adjustments.  

\begin{lemma}\label{lem:degenerate}
Let $k \geq 3$, and let $G$ be a connected graph with maximum degree $k$ distinct from $K_{k + 1}$. If $G$ is not $k$-regular, then a forest partition of $G$ can be found in $O(n + m)$ time. 
\end{lemma}

\begin{proof}
Since $G$ is connected and not $k$-regular, it is $(k - 1)$-degenerate. 
Let us first compute a $(k - 1)$-degenerate ordering  of the vertices of $G$ in $O(n + m)$ time as follows. 
We find a vertex $v$ of degree at most $k - 1$. Note that every neighbour of $v$ has degree at most $k - 1$ in $G - \{v\}$. 
So the recursive algorithm that consists of first deleting $v$ and then, for each vertex deleted, deleting all of its neighbours until the empty graph is obtained
gives a $(k - 1)$-degenerate ordering $v_1, v_2, \dots, v_n$ of $G$ in~$O(n + m)$ time. 

Let us now proceed to find a forest partition of $G$ in $O(n + m)$ time. 

\medskip

\textit{Case 1: $k$ is even} For $i=1,\ldots,n$, define $X_i = \{v_1, \dots, v_i\}$. By definition, $v_i$ has at most $k - 1$ neighbours in 
$X_{i-1}$. 
Let $r = \frac{k}{2}$.
It suffices to show that, for $2 \leq i \leq n$, we can compute in $O(1)$ time 
a partition $\{Y_1, \dots, Y_r\}$ of $X_i$, where $G[Y_{s}]$ is $1$-degenerate 
for $s = 1, \dots, r$, if we have as input such a partition of $X_{i-1}$. 
We note first that finding a partition of $X_1$ is trivial. Suppose $i > 1$ and
let $\{Z_1, \dots, Z_r\}$ be a partition of $X_{i-1}$ where $G[Z_s]$ is 
$1$-degenerate for $s = 1, \dots, r$. If $v_{i}$ has more than one neighbour 
in every $G[Z_{s}]$, then $v_i$ has at least $\sum_{i=1}^r2 = k$ 
neighbours in $X_{i-1}$, a contradiction. Hence, $v_i$ has at most one 
neighbour in at least one set $Z_q$, which we can find in $O(1)$ time
since we only need to check the neighbours of $v_i$ in $X_{i-1}$.
We put $v_i$ into $Z_q$ to get the desired partition for $X_i$ in $O(1)$ time. 

\medskip

\textit{Case 2: $k$ is odd} For $i=1,\ldots,n$, define $X_i = \{v_1, \dots, v_i\}$. By definition, $v_i$ has at most $k - 1$ neighbours in 
$X_{i-1}$. Let $r = \lceil \frac{k}{2} \rceil$.
It suffices to show that, for $2 \leq i \leq n$, we can compute in $O(1)$ time 
a partition $\{Y_1, \dots, Y_r\}$ of $X_i$, where $G[Y_1]$ is an independent set and $G[Y_{s}]$ is $1$-degenerate 
for $s = 2, \dots, r$, if we have as input such a partition of $X_{i-1}$. We note first that finding a partition of $X_1$ is trivial. Suppose $i > 1$ and
let $\{Z_1, \dots, Z_r\}$ be a partition of $X_{i-1}$ where $G[Z_1]$ is an independent set and $G[Z_s]$ is 
$1$-degenerate for $s = 2, \dots, r$. If $v_{i}$ has at least one neighbour in $G[Z_1]$ and more than one neighbour 
in every other $G[Z_{s}]$, then $v_i$ has at least $1 + \sum_{i=2}^r2 = k$ 
neighbours in $X_{i-1}$, a contradiction. Hence, $v_i$ has either no neighbour in $Z_1$ or at most one 
neighbour in at least one other set $Z_q$, which we can find in $O(1)$ time
since we only need to check the neighbours of $v_i$ in $X_{i-1}$.
We put $v_i$ into this set to get the desired partition for $X_i$ in $O(1)$ time. 
\end{proof}

A pair of
vertices $x$, $y$ in a connected graph $G$ is called an \emph{eligible pair} if $x$ and $y$ are at distance exactly two in $G$ and $G - \{x, y\}$ is connected. 
The proof of the next lemma makes use of the following result of Lov\'asz. 

\begin{lemma}[\cite{lovasz}]\label{lem:pair}
Let $G$ be a $2$-connected graph that is not complete or a cycle. Then an eligible pair of $G$ can be found in $O(n + m)$ time.  
\end{lemma}

\begin{lemma}\label{lem:2connected}
 Let $k \geq 3$, and let $G \not= K_{k+1}$ be a $2$-connected $k$-regular graph. Then a forest partition of $G$ can be found in $O(n + m)$ time. 
\end{lemma}

\begin{proof}
By Lemma~\ref{lem:pair}, we can find in $O(n + m)$ time an eligible pair of vertices $x$, $y$ in $G$. 
So there is a common neighbour of $x$ and~$y$ in $G$ 
that we denote $v$. Let $G'$ be the graph obtained from $G$ 
by identifying $x$ and $y$ into a new vertex $z$, and let $z_1, \dots, z_t$ denote
 the neighbours of $z$ distinct from $v$ that are common neighbours of $x$ and $y$ in~$G$. 
 
 \begin{claim}\label{claim:ordering}
 There is a $(k - 1)$-degenerate ordering of $G'$ that starts at $z$ that can be found in $O(n + m)$ 
time such that each $z_i$ has at most $k - 2$ neighbours earlier in the ordering. 
 \end{claim}

The proof of the claim is almost entirely contained in the proof \cite[Lemma 9]{feghalikempe}, but we repeat it for completeness.   
We shall prove Claim \ref{claim:ordering} by successively deleting vertices of $G'$ such that the earlier a vertex is deleted, the later it occurs in the ordering.  

The ordering starts at $z$ and ends at $v$ (note that $v$ has degree $k - 1$ in $G'$). 
The order of deletion of the remaining vertices is determined as follows. Since, by definition of an eligible pair, the graph $G^* = G' - \{z\}$ is connected, each neighbour of $z$ distinct from $v$ is joined to $v$ via a path in $G'$. We consider each such path (in an arbitrary order) and delete each (remaining) vertex of the path distinct from $v$ in the order in which it is encountered, if one traverses the path from the neighbour of $v$ towards the neighbour of $z$ on the path. Each vertex has degree at most $k - 1$ at the time it is deleted and each $z_i$ degree at most $k - 2$.
At this stage, we are left with a graph whose components are $(k - 1)$-degenerate. Then simply successively delete the remaining vertices 
distinct from $z$ of degree at most $k - 1$ in this graph. The claim is proved. 

Let us now find a forest partition $\F'$ of $G'$ in $O(n + m)$ time 
with the property that $z$ and $z_i$ belong to different forests for each $i= 1, \dots, t$. 
Define the sets $X_i$, $Y_i$ and $Z_i$ as in the proof of Lemma~\ref{lem:degenerate}. 
We  put $z \in Z_1$. Note that each $z_i$ has at 
most one neighbour in at least one $Z_q$ for some $q \geq 2$ 
since otherwise $z_i$ has at least $k - 1$ neighbours in $X_{i - 1}$, which contradicts Claim \ref{claim:ordering}.   
Thus, if we put $z_i \in Z_q$, we get $\F'$.    

To complete the proof, since every common neighbour of $x$ and $y$ in $G$ is not a member of $Z_1$, 
the graph $Z_1' = Z_1 \cup \{x, y\} \setminus \{z\}$ is also a forest and, if $k$ is odd, can be insured to be an independent set. Therefore, $\F = (\F' \setminus Z_1) \cup Z_1'$ is a forest partition of~$G$.  
\end{proof}

We are now able finish the proof of Theorem~\ref{thm:forests}.  

\begin{proof}[Proof of Theorem~\ref{thm:forests}]
By Lemma \ref{lem:split}, it suffices to show that we can find a forest partition of $G$ in $O(n + m)$ time.
We first check in $O(n + m)$ time whether $G$ is $k$-regular. 
If $G$ is not $k$-regular, we apply Lemma~\ref{lem:degenerate}. 
If $G$ is $k$-regular, we compute in $O(n + m)$ time a block decomposition of $G$ 
(by using, for example, a depth-first search algorithm). If $G$ is $2$-connected 
(that is, $G$ contains exactly one block), we can apply Lemma~\ref{lem:2connected}. 

So we can assume that $G$ is a connected $k$-regular graph and not $2$-connected. 
We consider an end block $B$ of $G$, and let~$v$ be the cut vertex of $G$ 
that is contained in $B$. Let $G' = G - B$, and let $B' = B - \{v\}$. Note that $G'$ and $B'$ are not $k$-regular. 
Applying Lemma~\ref{lem:degenerate}, we find a forest partition $\F'$ of $G'$ and a forest partition $\F''$ of $B'$. 

Two  cases arise. 

\medskip

\noindent
\textit{Case 1} There exists a forest $F' \in \F'$ and a forest $F'' \in \F''$ such that both $F'$ and $F''$ contain at least one neighbour of $v$. In this case, we pair off 
\begin{itemize}
\item $F'$ with $F''$,
\item the forests in $\F' \setminus F'$ with the forests in $\F'' \setminus F''$ arbitrarily, and
\item the independent set in $\F'$ with the independent set in $\F''$ (if $k$ is odd).
\end{itemize}

This yields a forest partition of $G - \{v\}$ that we denote $\F^*$. If $F'$ and $F''$ each contain exactly one neighbour of $v$, then $\F = (\F^* \setminus (F' \cup F'')) \cup (F' \cup F'' \cup \{v\})$ is a forest partition of $G$ that can be found in $O(n + m)$ time.  So we can assume that $F' \cup F''$ contains at least three neighbours of $v$. Suppose that $k$ is even. Since $|\F^*| = k/2$ and $v$ has degree exactly $k$, there exists a forest $F^* \in \F^* \setminus (F' \cup F'')$ that contains at most one neighbour of $v$. Hence $\F = (\F^* \setminus F^*) \cup (F^* \cup \{v\})$ is a forest partition of $G$. Similarly, if $k$ is odd, then $\F^*$ contains either a forest that contains at most one neighbour of $v$ or an independent set that does not contain a neighbour of $v$. In either case, a forest partition of $G$ can be found (in $O(n + m)$ time).  This completes Case 1. 

\medskip

\noindent
\textit{Case 2} $k$ is odd, the independent  set $I' \in \F'$ and the independent set $I'' \in \F''$ together contain at least two neighbours of $v$. In this case, we pair off
\begin{itemize}
\item $I'$ with $I''$ and
\item the forests in $\F'$ with the forests in $\F''$ arbitrarily.  
\end{itemize}

This yields a forest partition of $G - \{v\}$ that we denote $\F^*$. Since $|\F^*| = \lceil \frac{k}{2} \rceil$ and $v$ has degree exactly $k$, there must be a forest $F^* \in \F^*$ that contains at most one neighbour of $v$. Thus $\F = (\F^* \setminus F^*) \cup (F^* \cup \{v\}$ is a forest partition of $G$. This completes Case 2.  

From Cases 1 and 2, the one outstanding case to complete the proof of the theorem is when $k$ is odd and $v$ has: 
\begin{itemize}
\item precisely one neighbour in $G'$ that also belongs to the independent set in $\F'$, and
\item no neighbour in $B'$ that belongs to the independent set in $\F''$. 
\end{itemize}

In the remainder of the proof, we shall circumvent the second bullet point by constructing a new forest partition of $B'$ in $O(n + m)$ time whose independent set contains at least one neighbour of $v$. The following claim will be essential. 

\begin{claim}\label{claim:second}
There is a $(k - 1)$-degenerate ordering of the vertices of $B'$ that starts at some neighbour $w$ of $v$ that can be found in $O(n + m)$ time. 
\end{claim}

 Let $u$ be a vertex in $B \setminus (N(v) \cup \{v\})$.  
Since $B$ is an end block of $G$, it is $2$-connected. 
Thus, by Menger's Theorem, there are at least two internally disjoint paths in $B$ linking $u$ and $v$. 
Clearly, at least one of these paths contains some neighbour of $v$ distinct from $w$. We successively delete vertices of degree at most $k - 1$ in $B'$ starting with every neighbour of $v$ in $B'$ distinct from $w$ towards every other vertex distinct from~$w$. At the end of this procedure, we delete $w$. This proves the claim.

Using the ordering given by Claim \ref{claim:second}, we can now proceed (as in the proof of Lemma~\ref{lem:degenerate}) to obtain  a forest partition $\F''$ of $B'$ such that $w$ belongs to the independent set of $\F''$ (by simply placing $w \in Z_1$ at the start of the algorithm) in $O(n + m)$ time. 

Given that we have guaranteed that at least two vertices in the neighbourhood of $v$ belong to the independent set, the theorem follows.     
\end{proof}

\section{Hardness for large maximum degree}\label{section:hard}

In this section, we prove Theorem \ref{thm:complete}. This will be done by exhibiting polynomial time reductions from new variants of \textsc{SAT}, where each reduction ``corresponds'' to some combination of values of $p$ and $q$ in a $(p, q)$-partition of the graph. Let us first introduce these new variants of \textsc{SAT}.  

Recall that an instance $(X, \C)$ of \textsc{SAT} consists of a set of boolean variables $X = \{x_1, \dots, x_n\}$ and a collection of clauses $\C = \{C_1, \dots, C_m\}$, such that each clause
is a disjunction of literals, where a {\it literal} is either $x_i$ or its negation $\neg x_i$ for some $x_i \in X$. A function $g: X \rightarrow \{$true, false$\}$ is called a {\it satisfying truth assignment} if $\theta = C_1 \land \dots \land C_m$ is evaluated to true under $g$. A \textsc{SAT} instance $(X, \C)$
is called an \emph{\textsc{RSAT} instance} if each clause
is a disjunction of either exactly two literals or exactly four literals, and each literal
appears at most twice in $\C$.  A clause in $\C$ is called a {\it $k$-clause} for some positive integer $k$ if it contains exactly $k$ literals.  
We will reduce from the following variants of \textsc{RSAT}.

\problemdef{{\sc NAE-RSAT}}{An instance $(X, \C)$ of {\sc RSAT}.}{Does $(X, \C)$ have a satisfying truth assignment with at least one true literal and at least one false literal per clause?}

\problemdef{{\sc EXACT-RSAT}}{An instance $(X, \C)$ of {\sc RSAT}.}{Does $(X, \C)$ have a satisfying truth assignment with exactly one true literal per clause?}

\problemdef{{\sc ALL-RSAT}}{An instance $(X, \C)$ of {\sc RSAT}.}{Does $(X, \C)$ have a satisfying truth assignment with at least one true literal per 4-clause and exactly one true literal per 2-clause?}

\begin{lemma}\label{lem:sat}
Each of the above variants of \textsc{RSAT} is $\mathbb{NP}$-complete.
\end{lemma}

To prove the lemma, we require the following well-known $\mathbb{NP}$-complete decision problems; cf. Garey and Johnson \cite{garey}. 

An instance $(X, \C)$ of \textsc{SAT} is a \textsc{$4$-SAT} 
instance if every clause in $\C$ is a $4$-clause. 

\problemdef{{\sc $4$-SAT}}{An instance $(X, \C)$ of {\sc $4$-SAT}.}{Does $(X, \C)$ have a satisfying truth assignment?}

\problemdef{{\sc NAE $4$-SAT}}{An instance $(X, \C)$ of {\sc $4$-SAT}.}{Does $(X, \C)$ have a satisfying truth assignment with at least one true literal and at least one false literal per clause?
}

\problemdef{{\sc 1-in-$4$ SAT}}{An instance $(X, \C)$ of {\sc $4$-SAT}.}{Does $(X, \C)$ have a satisfying truth assignment with exactly one true literal per clause?}

\begin{proof}[Proof of Lemma \ref{lem:sat}]
Clearly, each of the above variants of \textsc{RSAT} are in $\mathbb{NP}$. We simultaneously show that they are $\mathbb{NP}$-hard by exhibiting a generic reduction from an instance $(X, \C)$ of \textsc{SAT} in which every clause contains exactly four literals.  (Our proof is identical to the proof that the variant of \textsc{3-SAT} in which every literal appears in at most two clauses is $\mathbb{NP}$-hard.  It will merely suffice to make a few additional observations.) 

Let $\theta = C_1 \land \dots \land C_m$.
If a variable $y \in X$ appears (as $y$ or $\neg y$) in $k > 1$ clauses, then we replace $y$ with a set of new variables $y^1, \dots, y^k$ in the following way: we replace the first occurrence of $y$ with $y^1$, the second occurrence of  $y$ with $y^2$, etc. If some of these occurrences are negated then we replace those occurrences with the negated versions of the new variables. We repeat this for each variable that appears in more than one clause. Next we link the new variables for $y$ to each other 
with a set of clauses $(y^1 \lor \neg y^2), (y^2 \lor \neg y^3), \dots, (y^k \lor \neg y^1)$. We denote by $(X', \C' = \{C'_1, \dots, C'_{m'}\})$ the resulting instance, 
and let $\theta' = C'_1 \land \dots \land C'_{m'}$. Notice that every literal appears in at most two clauses of $\C'$. 
Moreover, every 2-clause in $\C'$ has exactly one true literal since in every 
satisfying truth assignment $g'$ of the variables in $X'$, we have $g'(y^1) = \dots = g'(y^k)$ for every $y \in X$. Thus, for every satisfying truth 
assignment to the variables in $X$ and $X'$, 

\begin{itemize} 
\item $\theta$ has 
at least one true literal and at least one false literal per clause if and only if $\theta'$ has 
at least one true literal and at least one false literal per clause; 
\item $\theta$ has 
exactly one true literal per clause if and only if $\theta'$ has 
exactly one true literal per clause;
\item $\theta$ has at least one true literal per clause if and only if $\theta'$ has least one true literal per 4-clause and exactly one true literal per 2-clause. 
\end{itemize} 
Given that \textsc{4-SAT}, \textsc{NAE 4-SAT} and \textsc{1-in-4 SAT} are 
$\mathbb{NP}$-complete,  
it follows that \textsc{ALL-RSAT}, \textsc{NAE-RSAT} and \textsc{EXACT-RSAT} are $\mathbb{NP}$-hard. This completes the proof.   
\end{proof}

\tikzstyle{vertex}=[circle,draw=black, fill=black, minimum size=5pt, inner sep=1pt]
\tikzstyle{edge} =[draw,-,black,>=triangle 90]

\begin{figure}
\begin{center}
\begin{tikzpicture}[scale=0.75]

\foreach \pos/\name in {{(0, 2)/a1x}, {(1, 6)/x}, {(0, 10)/a2x}}
    \node[vertex] (\name) at \pos {};
    
\foreach \pos/\name in {{(10, 2)/a3x}, {(9, 6)/nx}, {(10, 10)/a4x}}
    \node[vertex] (\name) at \pos {};    

\foreach \num/\shift in  {1/a1x, 2/x, 3/a2x}{
\begin{scope}[shift={(\shift)}]
\foreach \pos/\name in {
{(3, 1)/t1\num},
{(1.5, 1)/k1\num},
{(1.5, 0)/k2\num},
{(1.5, -1)/k3\num},
{(3, -1)/t2\num}}
{
\node[vertex] (\name) at \pos {};

}

\foreach \source/ \dest in {\shift/k1\num, \shift/k2\num, \shift /k3\num, k2\num/k1\num, k2\num/k3\num, t1\num/k1\num, t1\num/k2\num, t1\num/k3\num, t2\num/k1\num, t2\num/k2\num, t2\num/k3\num}
\path[edge, black,  thick] (\source) --  (\dest);

\end{scope}
}

\foreach \num/\shift in  {4/a3x, 5/nx, 6/a4x}{
\begin{scope}[shift={(\shift)}]
\foreach \pos/\name in {
{(-3, 1)/t1\num},
{(-1.5, 1)/k1\num},
{(-1.5, 0)/k2\num},
{(-1.5, -1)/k3\num},
{(-3, -1)/t2\num}}
{
\node[vertex] (\name) at \pos {};
}

\foreach \source/ \dest in {\shift/k1\num, \shift/k2\num, \shift /k3\num, k2\num/k1\num, k2\num/k3\num, t1\num/k1\num, t1\num/k2\num, t1\num/k3\num, t2\num/k1\num, t2\num/k2\num, t2\num/k3\num}
\path[edge, black,  thick] (\source) --  (\dest);
\end{scope}
}

\node [below left=0] at (a1x) {$a_{x, 1}$};
\node [below left=0] at (a2x) {$a_{x, 2}$};
\node [below right=0] at (a3x) {$a_{x, 3}$};
\node [below right=0] at (a4x) {$a_{x, 4}$};
\node [below left=0] at (x) {$\widehat{x}$};
\node [below right=0] at (nx) {$\widetilde{x}$};

\node [above left=0] at (t11) {$a'_{x, 1}$};
\node [below left=0] at (t21) {$a^{*}_{x, 1}$};

\node [above left=0] at (t12) {$\widehat{x}'$};
\node [below left=0] at (t22) {$\widehat{x}^{*}$};

\node [above left=0] at (t13) {$a'_{x, 2}$};
\node [below left=0] at (t23) {$a^{*}_{x, 2}$};

\node [above right=0] at (t14) {$a'_{x, 3}$};
\node [below right=0] at (t24) {$a^{*}_{x, 3}$};

\node [above right=0] at (t16) {$a'_{x, 4}$};
\node [below right=0] at (t26) {$a^{*}_{x, 4}$};

\node [above right=0] at (t15) {$\widetilde{x}'$};
\node [below right=0] at (t25) {$\widetilde{x}^*$};

\path[edge, black,  thick] (t13) --  (t15);
\path[edge, black,  thick] (t23) --  (t25);

\path[edge, black,  thick] (t16) --  (t12);
\path[edge, black,  thick] (t26) --  (t22);

\path[edge, black,  thick] (t11) --  (t25);
\path[edge, black,  thick] (t22) --  (t14);

\path[edge, black,  thick] (t12) --  (t15);

\path[edge, black,  thick] (x) to[bend left=70]  (nx);

\path[edge, black,  thick] (x) to[bend right=10]  (t24);

\path[edge, black,  thick] (nx) to[bend left=10] (t21);

\path[edge, black,  thick] (k11) to[bend right=30] (k31);
\path[edge, black,  thick] (k12) to[bend right=30] (k32);
\path[edge, black,  thick] (k13) to[bend right=30] (k33);

\path[edge, black,  thick] (k14) to[bend left=30] (k34);
\path[edge, black,  thick] (k15) to[bend left=30] (k35);
\path[edge, black,  thick] (k16) to[bend left=30] (k36);

\end{tikzpicture}
\end{center}
\caption{The graph $S(x, 5)$.}\label{fig:variable}
\end{figure}
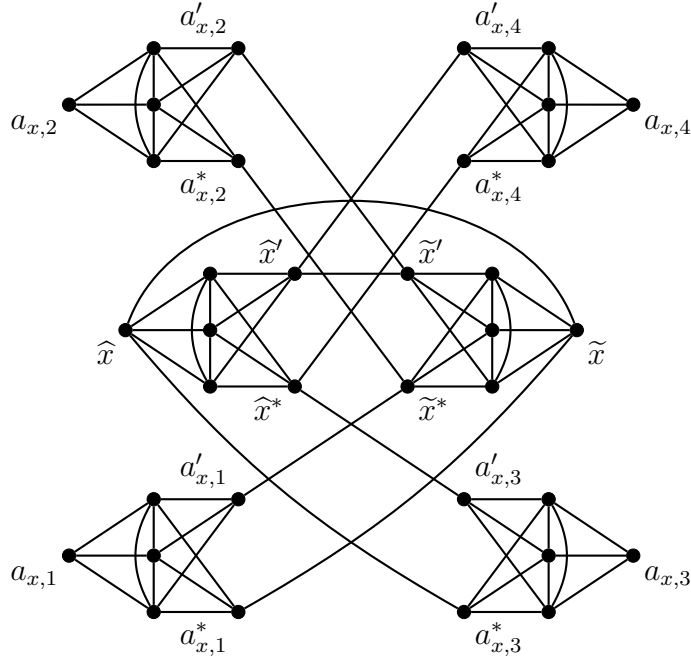


\begin{figure}
\begin{center}
\begin{tikzpicture}[scale=0.75]

\foreach \pos/\name in {{(0, 2)/x1}, {(4, 2)/x2}}
    \node[vertex] (\name) at \pos {};

\foreach \num/\shift in  {1/x1, 2/x2}{
\begin{scope}[shift={(\shift)}]
\foreach \pos/\name in {
{(1, 1)/k1\num},
{(0, 1)/k2\num},
{(-1, 1)/k3\num},
{(0, 2)/k4\num}}
{
\node[vertex] (\name) at \pos {};

}

\foreach \source/ \dest in {\shift/k1\num, \shift/k2\num, \shift /k3\num, k2\num/k1\num, k2\num/k3\num, k4\num/k3\num, k4\num/k2\num, k4\num/k1\num}
\path[edge, gray,  thin] (\source) --  (\dest);

\end{scope}
}
\path[edge, black,  thick] (x1) --  (x2);
\path[edge, black,  thick] (k41) --  (k42);

\path[edge, gray,  thin] (k11) to[bend right=30] (k31);
\path[edge, gray,  thin] (k12) to[bend right=30] (k32);

\node [below] at (x1) {$a_{x, f(\ell)}$};
\node [above] at (k41) {$a'_{x, f(\ell)}$};
\node [below] at (x2) {$a_{y, f(\ell')}$};
\node [above] at (k42) {$a'_{y, f(\ell')}$};

\end{tikzpicture}
\end{center}
\caption{The clause gadget that connects $S(x, 5)$ and $S(y, 5)$ via black edges, where literal $\ell$ corresponds to variable $x$ and literal $\ell'$ to variable $y$. 
Gray edges are edges of $S(x, 5)$ or $S(y, 5)$ and black edges are edges of $G$ and $H$. }\label{fig:clause1}


\end{figure}
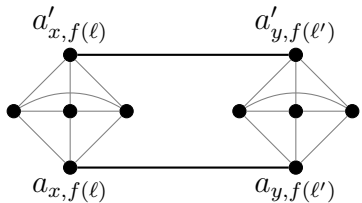

\begin{figure}
\begin{center}
\begin{tikzpicture}[scale=0.75]

\foreach \pos/\name in {{(4, 2)/x1}, {(4, 8)/x3}, {(-1, 4)/x2}, {(9, 4)/x4}}
    \node[vertex] (\name) at \pos {};
    
 \foreach \num/\shift in  {1/x1, 2/x2, 4/x4}{
\begin{scope}[shift={(\shift)}]
\foreach \pos/\name in {
{(1, 1)/k1\num},
{(0, 1)/k2\num},
{(-1, 1)/k3\num},
{(0, 2)/k4\num}}
{
\node[vertex] (\name) at \pos {};

}

\foreach \source/ \dest in {\shift/k1\num, \shift/k2\num, \shift /k3\num, k2\num/k1\num, k2\num/k3\num, k4\num/k3\num, k4\num/k2\num, k4\num/k1\num}
\path[edge, gray,  thin] (\source) --  (\dest);

\end{scope}
}

\begin{scope}[shift={(x3)}]
\foreach \pos/\name in {
{(1, -1)/k13},
{(0, -1)/k23},
{(-1, -1)/k33},
{(0, -2)/k43}}
{
\node[vertex] (\name) at \pos {};

}

\foreach \source/ \dest in {x3/k13, x3/k23, x3/k33, k23/k13, k23/k33, k43/k33, k43/k23, k43/k13}
\path[edge, gray,  thin] (\source) --  (\dest);


\end{scope}

\path[edge, black,  thick] (x1) to[bend left]  (k12);

\path[edge, black,  thick] (k12) to[bend left]  (x3);

\path[edge, black,  thick] (x3) to[bend left]  (k34);

\path[edge, black,  thick] (k34) to[bend left]  (x1);

\path[edge, black,  dashed] (k41) to  (k43);

\path[edge, black,  dashed] (k32) to[bend left=90] (k14);

\path[edge, gray,  thin] (k11) to[bend left=30] (k31);
\path[edge, gray,  thin] (k13) to[bend right=30] (k33);

\path[edge, gray,  thin] (x2) to[bend left=30] (k42);
\path[edge, gray,  thin] (x4) to[bend right=30] (k44);

\node [below] at (x1) {$a_{x_1, f(\ell_1)}$};
\node [right] at (k12) {$a_{x_2, f(\ell_2)}$};
\node [above] at (x3) {$a_{x_3, f(\ell_3)}$};
\node [left] at (k34) {$a_{x_4, f(\ell_4)}$};

\node [left] at (k41) {$a'_{x_1, f(\ell_1)}$};
\node [right] at (k43) {$a'_{x_3, f(\ell_3)}$};
\node [left] at (k32) {$a'_{x_2, f(\ell_2)}$};
\node [right] at (k14) {$a'_{x_4, f(\ell_4)}$};

\end{tikzpicture}
\end{center}
\caption{The clause gadget that connects $S(x_1, 5), \dots, S(x_4, 5)$ via black and dashed edges, where literal $\ell_i$ corresponds to variable $x_i$ for $1\leq i \leq 4$.
Gray edges are edges of $S(x_i, 5)$ for $1\leq i \leq 4$. Black edges are edges of both $G$ and $H$ while dashed edges are edges of $H$ only. }\label{fig:clause2}
\end{figure}
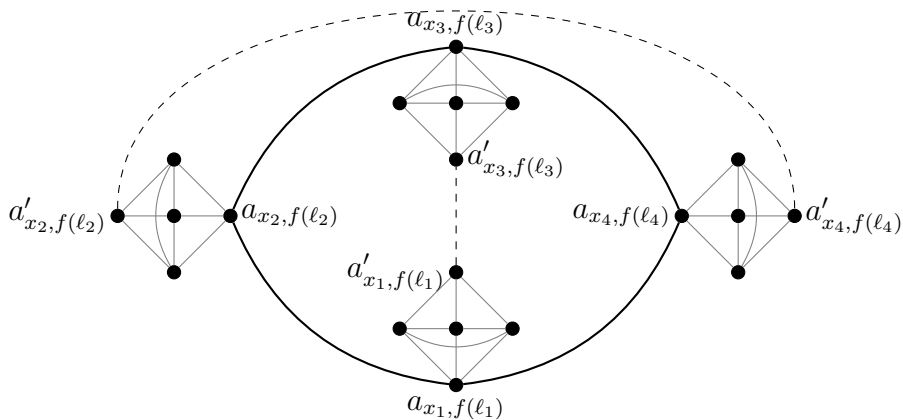

\begin{proof}[Proof of Theorem \ref{thm:complete}]
The problem is clearly in $\mathbb{NP}$. To show that it is $\mathbb{NP}$-hard, we simultaneously exhibit two reductions from a generic instance of \textsc{RSAT}. 

Given an instance $(X, \C)$ of \textsc{RSAT} we construct two graphs $G$ and $H$ of maximum degree $k \geq 5$ in the following way. We let $f: \bigcup_{C \in \C} C \rightarrow \{1, 2, 3, 4\}$ be a function 
that associates an integer between $1$ and $4$ to every literal such that
\begin{itemize}
\item  if literal $\ell$ is the $j$th occurrence of $y$ for some $y \in X$, then $f(\ell) = j$,  and
\item  if $\ell$ is the $j$th occurrence of $\neg y$ for some $y \in X$, then $f(\ell) = 2 + j$. 
\end{itemize}
 
Next, to each variable $x \in X$, we associate a variable gadget $S(x, k)$ as illustrated in Figure \ref{fig:variable} for $S(x, 5)$. It is   
a graph with $12$ vertices $a_{x, i}, a'_{x, i}, a^*_{x, i}$ for $i \in \{ 1, \dots, 4\}$, six vertices $\widehat{x}, \widehat{x}', \widehat{x}^*, \widetilde{x}, \widetilde{x}', \widetilde{x}^*$ and six copies $K^{(\widehat{x})}, K^{(\widetilde{x})}, K^{(1)}, K^{(2)}, K^{(3)}, K^{(4)}$ 
of a complete graph on $ k - 2$ vertices. We add edges between vertices $a_{x, i}, a'_{x, i}, a^*_{x, i}$ for $i = 1, \dots, 4$ and $\widehat{x}, \widehat{x}', \widehat{x}^*, \widetilde{x}, \widetilde{x}', \widetilde{x}^*$ in the way depicted in Figure \ref{fig:variable}. Moreover, we add edges from
\begin{itemize}
\item each of $ a_{x, i}, a'_{x, i}, a^*_{x, i}$ 
to every vertex of $K^{(i)}$, 
\item each of $\widehat{x}, \widehat{x}', \widehat{x}^*$ to every vertex of $K^{(\widehat{x})}$, and 
\item each of   $\widetilde{x}, \widetilde{x}', \widetilde{x}^*$ to every vertex of $K^{(\widetilde{x})}$. 
\end{itemize}

Next, for each 2-clause in $\C$ with literals $\ell$ and $\ell'$ corresponding, respectively, 
to variables $x$ and $y$ for some $x, y \in X$, we construct a 2-clause gadget that connects $S(x, k)$ and $S(y, k)$ in the way depicted in Figure \ref{fig:clause1}. We refer to $a_{x, f(\ell)}$ and $a_{y, f(\ell')}$ as \emph{special vertices} of the gadget. 
 Finally, for each 4-clause in $\C$ with literals $\ell_1, \dots, \ell_4$ corresponding, respectively, to variables $x_1$, $x_2$, $x_3$, $x_4$ for some $x_1, x_2, x_3, x_4 \in X$,
we construct a 4-clause gadget that connects $S(x_1, k), \dots, S(x_4, k)$ in the way depicted in Figure \ref{fig:clause2}. We also refer to each $a_{x_i, f(\ell_i)}$ as a \emph{special vertex} of the gadget. This completes the construction of $G$ and $H$. Note that $G$ and $H$ both have maximum degree $k \geq 5$. 


\begin{claim}\label{claim:variable}
Let $p, q \geq 0$ such that $p + q  = k - 3$. Then in every partition of $S(x, k)$ into a $p$-degenerate subgraph $P$ and a $q$-degenerate subgraph $Q$ either
\begin{itemize}\item $a_{x, 1}, a_{x, 2} \in V(P)$ and $a_{x, 3}, a_{x, 4} \in V(Q)$, or \item $a_{x, 1}, a_{x, 2} \in V(Q)$ and $a_{x, 3}, a_{x, 4} \in V(P)$. \end{itemize} 
\end{claim}

\begin{proof}[Proof sketch of claim]
It suffices to show that if $a_{x, 1} \in V(P)$, then  $a_{x, 2} \in V(P)$ and $a_{x, 3}, a_{x, 4} \in V(Q)$.
  Let us then assume that $a_{x, 1} \in V(P)$.  Recall that the set of neighbours of $a_{x, 1}$ in $S(x, k)$ induces a complete graph $K^{(1)}$ on $k - 2$ vertices.  Given that $p + q = k - 3$, exactly $p$ vertices of $K^{(1)}$ are members of $V(P)$ while the other $q + 1$ vertices of $K^{(1)}$ are members of $V(Q)$. This implies that both $a'_{x, 1}$ and $a^*_{x, 1}$ belong to $V(P)$.  
 
Let us now show that $\widetilde{x} \in V(Q)$. If $\widetilde{x} \in V(P)$, then again we find that $\widetilde{x'}$, $\widetilde{x^*}$ and exactly $p$ vertices of $K^{\widetilde{x}}$ are members of $V(P)$ while the other $q + 1$ of $K^{\widetilde{x}}$ are members of $V(Q)$. But the set $\{a_{x, 1}, \widetilde{x}, \widetilde{x'}, \widetilde{x^*}\} \cup ((K^{(1)} \cup K^{\widetilde{x}}) \cap V(P))$  induces a graph of minimum degree $p + 1$, which contradicts that $P$ is $p$-degenerate. Hence $\widetilde{x} \in V(Q)$.

Let us next show that $\widehat{x}, a_{x, 2} \in V(P)$. If $\widehat{x} \in V(Q)$ then, by the same reasoning, $\widehat{x^*}$, $\widehat{x'}$ and $q$ vertices of $K^{\widehat{x}}$ are members of $V(Q)$ while the other $p + 1$ vertices of $K^{\widehat{x}}$ are members of $V(P)$.  Similarly, since $\widetilde{x} \in V(Q)$, it follows that $\widetilde{x'}, \widetilde{x^*} \in V(Q)$. But the set $V(Q) \cap (\{\widehat{x}, \widehat{x^*}, \widehat{x'},  \widetilde{x}, \widetilde{x'}, \widetilde{x^*}\} \cup K^{\widetilde{x}} \cup K^{\widehat{x}})$ induces a graph of minimum degree $q + 1$, which contradicts that $Q$ is $q$-degenerate. Similarly, if $a_{x, 2} \in V(Q)$ then the set $V(Q) \cap (\{\widehat{x}, \widehat{x^*}, \widehat{x'},  a_{x, 2}, a'_{x, 2}, a^*_{x, 2}\} \cup K^{(2)} \cup K^{\widehat{x}})$ induces a graph of minimum degree $q + 1$. Hence $\widehat{x}, a_{x, 2} \in V(P)$ as needed.

It remains to show that $a_{x, 3}, a_{x, 4} \in V(Q)$. Using the fact that $\widehat{x} \in V(P)$, one can argue as before that  $a_{x, 3}$ and $a_{x, 4}$ are indeed both members of $V(Q)$.
 \end{proof}
 
 \begin{claim}\label{claim:f}
 Let $p, q \geq 0$ such that $p + q  = k - 3$. Consider any partition of $S(x, k)$ into a $p$-degenerate subgraph $P$ and a $q$-degenerate subgraph $Q$. Then each vertex in $P$ (respectively, $Q$) that is not a special vertex or a neighbour of a special vertex has degree $p$ in $P$ (respectively, $q$ in $Q$).
 \end{claim}
 
 \begin{proof}
 This follows from the proof of Claim \ref{claim:variable}. 
 \end{proof}
 
 We distinguish three cases depending on the values of $p$ and $q$. 
 
 \medskip
 
 \noindent
\textit{ Case 1} $p = 1$ and $q \geq 2$. 
 
 \medskip
 
In this case, we reduce from \textsc{ALL-RSAT}. More precisely, we will show that $(X, \C)$ has a satisfying truth assignment with exactly one true literal per 2-clause if and only if $G$ admits a partition into a $p$-degenerate graph $P$ and a $q$-degenerate graph $Q$. By Lemma \ref{lem:sat}, deciding whether $G$ has a $(p, q)$-partition with $p = 1$ and $q \geq 2$ is $\mathbb{NP}$-hard.   
 
Suppose that $(X, \C)$ has a satisfying truth assignment with exactly one true literal per 2-clause. For each $x \in X$ and each literal $\ell$ corresponding to $x$, if $\ell$ is set to true, then we put $a_{x, f(\ell)}$ in $V(Q)$, and if $\ell$ is set to false, then we put $a_{x, f(\ell)}$ in $V(P)$.  
  
  By Claim \ref{claim:variable}, this partial $(p, q)$-partition of $G$ extends to a $(p, q)$-partition of each variable gadget of $G$. To see that this gives a partition of $G$ into a $p$-degenerate graph $P$ and a $q$-degenerate $Q$, notice that
  \begin{itemize}
 \item  the degrees of vertices in each of $P$ and $Q$ are not affected by the 2-clause gadgets, and
 \item  no cycle in $P$ is formed by the $4$-clause gadgets given that at least one special vertex of each $4$-clause gadget belongs to $Q$.
  \end{itemize}    
 Using Claim \ref{claim:f}, one may then easily check that a $p$-degenerate ordering of the vertices in $P$ and a $q$-degenerate ordering of the vertices in $Q$ can be obtained.
        
  Conversely, suppose that $G$ admits a partition into a $p$-degenerate graph $P$ and a $q$-degenerate graph $Q$.  For each $x \in X$ and each literal $\ell$ corresponding to $x$,  if $a_{x, f(\ell)} \in V(Q)$,  then we set $\ell$ to true, and if $a_{x, f(\ell)} \in V(P)$, then we set $\ell$ to false.
  
   By Claim \ref{claim:variable}, this is a valid truth assignment to the variable in $X$.  Notice that at least one special vertex of each 4-clause gadget is a member of $V(Q)$ given that $p = 1$. Notice also that exactly one special vertex of each 2-clause gadget is a member of $V(Q)$, for if two special vertices, say $a_{x, i}$ and $a_{x, j}$, of a 2-clause gadget are in $V(Q)$, then $a_{x, i}$, $a_{x, j}$, their neighbours in $Q$ and $a'_{x, i}$ and $a'_{x, j}$ would induce a a graph of minimum degree $q + 1$ in $Q$. Similarly, if both $a_{x, i}$ and $a_{x, j}$ are in $V(P)$, then  these vertices together with their neighbours in $P$ and $a'_{x, i}$ and $a'_{x, j}$ would induce a graph of minimum degree $2$ in $P$. Hence we have a satisfying truth assignment of $(X, \C)$ such that each $4$-clause has at least one true literal and each $2$-clause exactly one true literal. This completes Case 1.

 \bigskip
 
 \noindent
 \textit{Case 2} $p = 0$ and $q \geq 2$. 
 
 \medskip
 
In this case, we reduce from \textsc{EXACT-RSAT}. More precisely, we will show that $(X, \C)$ has a satisfying truth assignment with exactly one true literal per clause if and only if $H$ admits a partition into a $p$-degenerate graph $P$ and a $q$-degenerate graph $Q$. By Lemma \ref{lem:sat}, deciding whether $H$ has a $(p, q)$-partition with $p = 0$ and $q \geq 2$ is $\mathbb{NP}$-hard. 
 
 Suppose that $(X, \C)$ has a satisfying truth assignment with exactly one true literal per clause. For each $x \in X$ and each literal $\ell$ corresponding to $x$, if $\ell$ is set to true, then we put $a_{x, f(\ell)}$ in $V(P)$, and if $\ell$ is set to false, then we put $a_{x, f(\ell)}$ in $V(Q)$. By Claim \ref{claim:variable}, this partial partition of $H$ can be extended to a $(p, q)$-partition of each variable gadget. Clearly, this forms a $(p, q)$-partition of every 2-clause gadget.  To see that it also forms a $(p, q)$-partition of every 4-clause gadget (and therefore a $(p, q)$-partition of $H$),  consider a 4-clause gadget with special vertices   $a_{x_1, f(\ell_1)}, \dots, a_{x_4, f(\ell_4)}$. Suppose without loss of generality that $a_{x_1, f(\ell_1)}$ are in $V(P)$ and $a_{x_2, f(\ell_2)}, a_{x_3, f(\ell_3)}, a_{x_4, f(\ell_4)}$ are in $V(Q)$. We extend this partition to the rest of the 4-clause gadget so that $a'_{x_1, f(\ell_1)} \in V(P)$ and $a'_{x_2, f(\ell_2)}, a'_{x_3, f(\ell_3)}, a'_{x_4, f(\ell_4)} \in V(Q)$. It is clear that no two vertices in $V(P)$ are adjacent so it remains to show that vertices in $V(Q)$ induce a $q$-degenerate subgraph. Notice that vertex $a'_{x_4, f(\ell_4)}$ has $q$ neighbours in $V(Q)$ since its neighbours that are not in $K^{(4)}$ are in $V(P)$.  The procedure of first deleting $a'_{x_4, f(\ell_4)}$, followed by the neighbours of $a'_{x_4, f(\ell_4)}$ in $Q$, followed by $a_{x_4, f(\ell_4)}$, $a_{x_2, f(\ell_2)}$  and $a_{x_3, f(\ell_3)}$ in this order etc. (the rest of details are left to the reader) a $q$-degenerate ordering of vertices in $Q$ can be obtained.  
 
 Conversely, suppose that $H$ has a $(p, q)$-partition. For each $x \in X$ and each literal $\ell$ corresponding to $x$,  if $a_{x, f(\ell)}$ is in $V(P)$,  we set $\ell$ to true, and if $a_{x, f(\ell)}$ is in $V(Q)$, we set $\ell$ to false. By Claim \ref{claim:variable}, this is a valid truth assignment to the variable in $X$. As in Case 1, exactly one special vertex of each 2-clause is in $V(P)$. Consider a 4-clause with special vertices   $a_{x_1, f(\ell_1)}, \dots, a_{x_4, f(\ell_4)}$. Exactly one of these special vertices is in $V(P)$: 
 
 \begin{itemize}
\item If at least two of them are in $V(P)$, then they are not adjacent (since $p = 0$). Thus  $a_{x_{i}, f(\ell_{i})}, a_{x_{i + 2}, f(\ell_{i+2})} \in V(P)$ (for some $i = 1, 2$), which  implies  $a'_{x_{i}, f(\ell_{i})}, a'_{x_{i+2}, f(\ell_{i + 2})} \in V(P)$.  This is impossible since  $a'_{x_{i}, f(\ell_{i})}$ and $a'_{x_{i+2}, f(\ell_{i + 2})}$ are adjacent. 
\item If all of them are members of $V(Q)$, then, considering edges of the gadget that are in $H$ but not in $G$, one can find a subgraph of $Q$ with minimum degree $q + 1$. 
 \end{itemize}
This shows that we have a satisfying truth assignment of $(X, \C)$ with exactly one true literal per clause. This completes Case 2.

\bigskip

\noindent
\textit{Case 3} $p, q \geq 2$. 

\medskip

In this case, we reduce from \textsc{NAE-RSAT}. More precisely, we must show that a $(X, \C)$ has a satisfying truth assignment with at least one false literal and at least one true literal per clause if and only if $H$ admits a $(p, q)$-partition. By Lemma \ref{lem:sat}, deciding whether $H$ has a $(p, q)$-partition with $p, q \geq 2$ is $\mathbb{NP}$-hard. Since the arguments are entirely similar to those of Cases 1 and 2, we leave the details to the reader.

The proof of the theorem is complete. 
\end{proof}

\section{Concluding remarks}

Let us first note that a straightforward adaptation of the proof of Lemma \ref{lem:degenerate} leads to the following statement. 
  
  \begin{proposition}\label{prop:deg}
 Given integers $s \geq 2$ and $p_1, \dots, p_s \geq 0$, a $(p_1, \dots, p_s)$-partition of a graph $G$ with maximum degree $k \geq 3$ that is not $k$-regular can be found in $O(n + m)$-time as long as $\sum_{i=1}^s p_i \geq k - s$.
  \end{proposition}

  It is therefore unlikely that the time complexity increases by more than a factor of $n$ in the outstanding case where $G$ is $k$-regular. 

\begin{conjecture}
Let $G$ be a connected graph with maximum degree $k \geq 3$ distinct from $K_{k+1}$. For every $s \geq 2$ and $p_1,\ldots,p_s\geq 0$ such that $\sum_{i = 1}^sp_i \geq k-s$, a $(p_1, \dots, p_s)$-partition of $G$ can be found in~$O(n^2)$ time.
\end{conjecture}

We make a few remarks on the case $s \geq 3$ with $\sum_{i=1}^s p_i < k - s$.  A simple application of Proposition~\ref{prop:deg} leads to the following statement.
\begin{proposition}\label{prop:f}
Given non-negative integers $p, q, p_1, p_2, \dots, p_t, q_1, \dots, q_{t'}$ such that $\sum p_i = p - t$ and $\sum q_i = q - t'$, 
if a graph $G$ is $(p, q)$-partitionable, then $G$ is also $(p_1, \dots, p_t, q_1, \dots, q_{t'})$-partitionable.
\end{proposition} 

Proposition \ref{prop:f}  can be understood to mean (although rather imprecisely) that the complexity of Problem \ref{problem1} in the situation when $\sum_{i=1}^s p_i < k - s$ does not increase as $s$ increases. Phrased differently, Proposition \ref{prop:f} states informally that if one can find a partition into two subgraphs with prescribed degeneracy, then one can also find a partition of the same graph into more than two subgraphs with prescribed degeneracy, provided some condition on the sum of the prescribed degeneracies is met.

  We therefore hoped that the problem is $\mathbb{NP}$-complete whenever $s$ is as large as possible (that is, when a partition into independent sets is sought) as this would suggest that the problem is $\mathbb{NP}$-complete for every $s \geq 2$. As it happens, however, when $s$ is of maximum value, the problem is tractable as long as $k$ is not too small and $(k - s) - \sum_{i=1}^s p_i$ is not very large~\cite{molloy}.  This might indicate that determining the frontier between tractability and hardness for every value of $s$ in Problem \ref{problem1}  will be a difficult task. 
 
 \section*{Acknowledgements}
 
 The authors are grateful to both referees for their careful reading of the paper and for their suggestions that significantly improved the presentation of the paper. This work received support from the Research Council of Norway via the project CLASSIS, grant number 249994.

 \bibliography{bibliography}{}

\begin{thebibliography}{10}

\bibitem{alon}
N.~Alon, J.~Kahn, and P.~D. Seymour.
\newblock Large induced degenerate subgraphs.
\newblock {\em Graphs and Combinatorics}, 3(1):203--211, 1987.

\bibitem{wood}
B.~Baetz and D.~R. Wood.
\newblock Brooks' vertex-colouring theorem in linear time.
\newblock {\em CoRR}, abs/1401.8023, 2014.

\bibitem{bonamy}
M.~Bonamy, K.~K. Dabrowski, C.~Feghali, M.~Johnson, and D.~Paulusma.
\newblock Recognizing graphs close to bipartite graphs.
\newblock {\em Proc. MFCS 2017, LIPIcs}, 83:70:1--70:14, 2017.

\bibitem{bonamybipartite}
M.~Bonamy, K.~K. Dabrowski, C.~Feghali, M.~Johnson, and D.~Paulusma.
\newblock Recognizing graphs close to bipartite graphs with an application to
  colouring reconfiguration.
\newblock {\em CoRR}, abs/1707.09817, 2017.

\bibitem{toft}
O.~V. Borodin, A.~V. Kostochka, and B.~Toft.
\newblock Variable degeneracy: extensions of {Brooks'} and {Gallai's} theorems.
\newblock {\em Discrete Mathematics}, 214(1--3):101--112, 2000.

\bibitem{brooks}
R.~L. Brooks.
\newblock On colouring the nodes of a network.
\newblock {\em Mathematical Proceedings of the Cambridge Philosophical
  Society}, 37(2):194--197, 1941.

\bibitem{arboricity1}
P.~A. Catlin and H.-J. Lai.
\newblock Vertex arboricity and maximum degree.
\newblock {\em Discrete Mathematics}, 141(1--3):37--46, 1995.

\bibitem{FJP16}
C.~Feghali, M.~Johnson, and D.~Paulusma.
\newblock A reconfigurations analogue of {Brooks'} {Theorem} and its
  consequences.
\newblock {\em Journal of Graph Theory}, 83(4):340--358, 2016.

\bibitem{feghalikempe}
C.~Feghali, M.~Johnson, and D.~Paulusma.
\newblock Kempe equivalence of colourings of cubic graphs.
\newblock {\em European Journal of Combinatorics}, 59:1--10, 2017.

\bibitem{garey}
M.~R. Garey and D.~S. Johnson.
\newblock {\em Computers and intractability}, volume~29.
\newblock wh freeman New York, 2002.

\bibitem{arboricity2}
H.~V. Kronk and J.~Mitchem.
\newblock Critical point-arboritic graphs.
\newblock {\em Journal of the London Mathematical Society}, 2(3):459--466,
  1975.

\bibitem{lick}
D.~R. Lick and A.~T. White.
\newblock k-degenerate graphs.
\newblock {\em Canadian J. of Mathematics}, 22:1082--1096, 1970.

\bibitem{lovasz}
L.~Lov\'asz.
\newblock Three short proofs in graph theory.
\newblock {\em Journal of Combinatorial Theory, Series B}, 19:269--271, 1975.

\bibitem{molloy}
M.~Molloy and B.~Reed.
\newblock Colouring graphs whose chromatic number is almost their maximum
  degree.
\newblock In {\em Latin American Symposium on Theoretical Informatics}, pages
  216--225. Springer, 1998.

\bibitem{Tho22}
C.~Thomassen.
\newblock Decomposing a planar graph into degenerate graphs.
\newblock {\em Journal of combinatorial theory, Series B}, 65(2):305--314,
  1995.

\bibitem{Th01}
C.~Thomassen.
\newblock Decomposing a planar graph into an independent set and a 3-degenerate
  graph.
\newblock {\em Journal of Combinatorial Theory, Series B}, 83(2):262--271,
  2001.

\bibitem{wu}
Y.~Wu, J.~Yuan, and Y.~Zhao.
\newblock Partition a graph into two induced forests.
\newblock {\em Journal of Mathematical Study}, 29:1--6, 1996.

\bibitem{yang}
A.~Yang and J.~Yuan.
\newblock Partition the vertices of a graph into one independent set and one
  acyclic set.
\newblock {\em Discrete Mathematics}, 306(12):1207--1216, 2006.

\end{thebibliography}
\bibliographystyle{abbrv}
 
\end{document}